\documentclass[letterpaper, 10 pt, conference]{ieeeconf}  

\IEEEoverridecommandlockouts                              
\usepackage{algorithm,algcompatible}
\usepackage{comment} 
\usepackage{lipsum} 
\usepackage{graphics} 
\usepackage{epsfig} 
\usepackage{times} 
\usepackage{amsmath} 
\usepackage{amssymb}  
\usepackage{color}
\usepackage{xcolor}
\usepackage{caption}
\usepackage{subcaption}
\usepackage{mathabx}
\newcommand{\norm}[1]{\left\lVert#1\right\rVert}
\newtheorem{theorem}{Theorem}

\newtheorem{lemma}{Lemma}

\newtheorem{assumption}{Assumption}

\title{\LARGE \bf
Momentum-based Accelerated Q-learning
}

\author{Bowen Weng, Lin Zhao, Huaqing Xiong and Wei Zhang
\thanks{This work was supported in part by the National Science Foundation under grant CNS-1552838.}
\thanks{Bowen Weng and Huaqing Xiong are with the Department of Electrical and Computer Engineering, The Ohio State University, Columbus, OH 43210.
        {\tt\small weng.172, xiong.309}}%
\thanks{Lin Zhao is with the Pittsburgh Technology Center, Aptiv PLC.
        {\tt\small zhao.833@osu.edu}}%
\thanks{Wei Zhang is with the Department of Mechanical and Energy Engineering, Southern University of Science and Technology (SUSTech), Shenzhen, China. 
        {\tt\small zhangw3@sustech.edu.cn}}%
}

\begin{document}

\maketitle
\thispagestyle{empty}
\pagestyle{empty}

\begin{abstract}
This paper studies accelerated algorithms for Q-learning. We propose an acceleration scheme by incorporating the historical iterates of the Q-function. The idea is conceptually inspired by the momentum-based acceleration methods in the optimization theory. Under finite state-action space settings, the proposed accelerated Q-learning algorithm provably converges to the global optimum with a rate of $\mathcal{O}(1/\sqrt{T})$. While sharing a comparable theoretic convergence rate with the existing Speedy Q-learning (SpeedyQ) algorithm, we numerically show that the proposed algorithm outperforms SpeedyQ via playing the FrozenLake grid world game. 
Furthermore, we generalize the acceleration scheme to the continuous state-action space case where function approximation of the Q-function is necessary. In this case, the algorithms are validated using commonly adopted testing problems in reinforcement learning, including two discrete-time linear quadratic regulation (LQR) problems from the Deepmind Control Suite, and the Atari 2600 games. Simulation results show that the proposed accelerated algorithms can improve the convergence performance compared with the vanilla Q-learning algorithm.
\end{abstract}

\section{INTRODUCTION}

Reinforcement learning (RL) aims to study how an agent learns a policy through interacting with its environment to minimize the accumulative loss for a task. RL has received dramatically growing attention and gained success in various tasks, such as playing video games~\cite{mnih2013playing}, bipedal walking~\cite{castillo2018reinforcement} and studying control systems~\cite{lewis2012reinforcement}, to name a few. This paper focuses on the Q-learning algorithm which is a model-free RL algorithm to find an estimate of the optimal action value function. 

Ever since the first proposal of the Q-learning algorithm in 1989~\cite{watkins1992q}, the method has been studied extensively in the finite state-action space. When the size of the state-action space is relatively small, the Q-function can be explicitly represented as a tabular function leading to a convenient proof of convergence~\cite{jaakkola1994convergence}. 

When the state-action space is continuous or considerably large, Q-learning usually requires function approximations. Lewis et. al. studied the Q-learning problem for linear control systems~\cite{al2007model} and extended the method to the continuous time domain~\cite{vamvoudakis2017q,vrabie2009adaptive}. They consider value iteration with appropriate sampling and customized Q-function structure. The step of {\em target update} was later introduced in the Deep Q-Network (DQN) learning~\cite{mnih2013playing} with the Q-function being parameterized as a deep neural network. DQN has gained great success in playing video games~\cite{oh2015action} that significantly exceeds human-level of performance. This also leads to various improved algorithms for Q-learning~\cite{van2016deep} and for general RL~\cite{schulman2017proximal}. 

Besides the exploration on the improved Q-learning algorithms with better performance in applications, another line of research lies in the convergence analysis of variants of the Q-learning algorithms~\cite{azar2011speedy,beck2012error,bertsekas1996neuro,devraj2017fastest,even2003learning}. Given that the training speed largely determines how an algorithm can contribute to the real application, accelerating the convergence is always of great interest. Optimization theory has provided effective schemes of acceleration with theoretic guarantees. One of the most popular schemes is based on the so-called momentum idea by involving more historical information into the update. Momentum-based algorithms, including Heavy-ball (HB)~\cite{polyak1964some}, Nesterov's accelerated gradient (NAG)~\cite{nesterov2013introductory}, have been proved to be able to accelerate the convergence when loss functions are strongly convex~\cite{ghadimi2015global,nesterov2013introductory}. Under general convex or nonconvex settings, the acceleration of these algorithms has not been established theoretically. Nevertheless, its convergence can still be guaranteed for specific classes of nonconvex loss functions~\cite{xiong2018analytical} and numerical results also show great success~\cite{attouch2017rate,dozat2016incorporating}. A successful application of the momentum to accelerate Q-learning is the so-called Speedy Q-learning (SpeedyQ)~\cite{azar2011speedy}. It is provably better than the vanilla Q-learning assuming a finite state-action space.

Our contribution in this paper is twofold. First, we propose a new accelerated Q-learning scheme which is inspired by the general momentum-based optimization algorithms. Under the finite state-action space and the synchronous sampling settings~\cite{even2003learning} , we prove the upper bounds of the convergence rate which is comparable to that of SpeedyQ. Furthermore, using a popular grid world game, we numerically show that the proposed acceleration scheme can outperform SpeedyQ under the same settings as considered in the theoretical derivations. Second, we also generalize our acceleration scheme to the continuous state-action space case, where the Q-function is usually approximated by parametric functions. In this case, we numerically evaluate the proposed algorithms in various challenging tasks, including two linear quadratic regulation problems from the Deepmind Control Suite~\cite{tassa2018deepmind} and the Atari 2600 video games. A significant improvement of the performance over the vanilla Q-learning is shown by the simulation results.


The rest of the paper is organized as follows. Section~\ref{sec: tabularQ} introduces the background of Q-learning and SpeedyQ. Section~\ref{sec:AQL} proposes a new acceleration scheme, followed by the convergence analysis and numerical performance comparison with SpeedyQ. In Section~\ref{sec:ctnQ}, we generalize our acceleration scheme to the case where the state-action space is continuous or considerably large. We also provide numerical results to show the promising applications of our algorithms in more complicated applications.

\section{Preliminaries}\label{sec: tabularQ}

In this section, we provide the background of Q-learning. We also briefly revisit the SpeedyQ algorithm for comparison in later sections.

\subsection{Q-learning}
We consider the standard reinforcement learning settings, where a learning agent (e.g. controller or control policy) interacts with a (possibly stochastic) environment (e.g. process or system dynamics, etc.). This interaction is usually modeled as a discrete-time discounted Markov Decision Processes (MDPs), described by a quintuple $(\mathcal{X},\mathcal{U},P,R, \gamma)$, where $\mathcal{X}$ is the state space, $\mathcal{U}$ is the action space, $P:\mathcal{X}\times \mathcal{U} \times \mathcal{X}\mapsto [0,1]$ is the probability kernel for the state transitions, e.g., $P(\cdot|x, u)$ denotes the probability distribution of the next state given current state $x$ and action $u$. In addition, $R: \mathcal{X}\times \mathcal{U}\mapsto[0,R_{\max}]$ is the reward function (or negative of the cost function) mapping station-action pairs to a bounded subset of $\mathbb{R}$, and $\gamma\in (0,1)$ is the discount factor. The optimal stationary policy $\pi^{\star}: \mathcal{X}\mapsto \mathcal{U}$ of MDP is defined as the solution of following optimization problem:
\begin{align}
    & \underset{\pi}{\text{maximize}}
    & & J_{\pi}(x_0) = {\mathbb{E}_P}\left\{\sum_{k=0}^{\infty} \gamma^k R(x_k, \pi(x_k))\right\}, \nonumber\\
    & \text{subject to}
    & & x_{k+1} \sim P(\cdot|x_k, \pi(u_k)),\label{eq:systemEquation}
\end{align}
where $\mathbb{E}_P$ denotes the expectation with respect to the transition probability $P$.
The above optimization problem seeks to maximize the expected accumulated discounted rewards over different policies $\pi$.

A stationary policy $\pi$ induces a Q-function $Q^\pi$ which satisfies the \textit{Bellman equation}:
\[
Q^{\pi}(x,u):=R(x,u)+\gamma\mathbb{E}_P Q^{\pi}(x^{\prime},\pi(x^{\prime})),
\]
where $x^{\prime}\sim P(\cdot|x,u)$ denotes the next state.
 
The Bellman operator $\mathcal{T}$ is defined pointwisely as
\begin{equation}
    \label{eq:BellmanOperator}
    \mathcal{T}Q(x,u) = R(x,u)+\gamma\mathbb{E}_P \underset{u^\prime \in U(x^\prime)}{\text{max}}Q((x^\prime,u^\prime),
\end{equation}
which can be shown to be a contractive in the supremum norm (i.e., $\norm{Q}:=\sup_{x,u}|Q(x,u)|$) 
\begin{equation}
    \label{eq:Contraction}
    \norm{\mathcal{T}Q(x,u) - \mathcal{T}Q'(x,u)} \leq \gamma\norm{Q(x,u) - Q'(x,u)},
\end{equation}
and its unique fixed point is the optimal Q-function $Q^{\star}$, i.e., $\mathcal{T}Q^{\star}(x,u)=Q^{\star}(x,u)$, which also satisfies the  
 \textit{optimal Bellman equation}~\cite{bertsekas1996neuro}: 
\begin{equation}
    \label{eq:OptimalBellmanEq}
    Q^{\star}(x,u) = R(x,u)+\gamma\mathbb{E}_P \underset{u^\prime \in U(x^\prime)}{\text{max}}Q^{\star}(x^\prime,u^\prime),
\end{equation}
Therefore, starting with an arbitrary Q-function, we can apply the Bellman operator $\mathcal{T}$ iteratively to learn $Q^{\star}$. 

Let $J^{\star}(x):=J_{\pi^{\star}}(x)$ be the optimal value function when applying the optimal policy $\pi^{\star}$. It relates to $Q^{\star}$ as follows
\begin{equation}
    \label{eq:bellQ}
    J^{\star}(x)=\underset{u\in U(x)}{\text{max}}Q^{\star}(x,u), \forall x\in\mathcal{X},
\end{equation}
where $U(x)$ denotes the admissible set of actions at state $x$. Hence, the optimal policy can be obtained from the optimal Q-function as:
\begin{equation}
    \label{eq:optPol}
    \pi^{\star}(x) = \underset{u \in U(x)}{\text{argmax}}\ Q^{\star}(x,u),\forall x\in\mathcal{X}
\end{equation}
Note that the knowledge of the transition probability $P$ is not needed in~\eqref{eq:optPol}, which is the advantage of Q learning.

In practice, exact evaluation of the Bellman operator~\eqref{eq:BellmanOperator} is usually infeasible due to the lack of the knowledge of the system dynamics (i.e. the transition probabilities). Instead, the \textit{empirical Bellman operator} is evaluated using samples~\cite{jaakkola1994convergence}. 
Specifically, for the $k$th round of iteration at state-action pair $(x,u)$, we sample the next state $y_k\sim P(\cdot|x,u)$, and then evaluate the empirical Bellman operator $\mathcal{T}_{k}$ as
\begin{equation}\label{eq:empTQ}
    \mathcal{T}_{k}Q_k(x,u) = R(x,u)+\gamma \underset{u'\in U(y_k)}{\max}Q_k(y_k,u'),
\end{equation}
where note that the subscript of $\mathcal{T}_k$ is to track that of samples $y_k$. As an example, the vanilla Q-learning is implemented as
\begin{equation}\label{eq:tabular q learning}
    Q_{k+1} =  Q_k  - \alpha_k(Q_k - \mathcal{T}_{k}Q_k),
\end{equation}
where $\alpha_k$ is the step size and we omit the dependence on $(x,u)$ hereafter when no confusion can arise. 



\subsection{Speedy Q-learning}
In optimization and deep learning, momentum-based schemes, including Heavy-ball (HB)~\cite{polyak1964some} and Nesterov's accelerated gradient (NAG)~\cite{nesterov2013introductory}, have been widely used to accelerate the convergence of gradient based algorithms. Such schemes also inspired some improved Q-learning algorithms such as SpeedyQ~\cite{azar2011speedy}, which follows the update as:
\begin{equation}\label{eq:sql}
    Q_{k+1} = Q_k + \alpha_k(\mathcal{T}_{k}Q_k - Q_k) + (1-\alpha_k)(\mathcal{T}_{k}Q_k - \mathcal{T}_{k}Q_{k-1}),
\end{equation}
where $\alpha_k=\frac{1}{k+1}$. Compared with~\eqref{eq:tabular q learning}, SpeedyQ added a momentum term $\mathcal{T}_{k}Q_k - \mathcal{T}_{k}Q_{k-1}$. This is a straightforward setup considering the history momentum. In the following section, we further explore the accelerated Q-learning framework with a more sophisticated design.

\section{Accelerated Q-learning}\label{sec:AQL}
In this section, we propose a new class of accelerated Q-learning (AQL) algorithms inspired by general momentum-based optimization algorithms.
A generic form of AQL is given by
\begin{equation}
    \begin{aligned}
    \label{eq: acc tabular q learning}
    & S_k = (1-a_k)Q_{k-1} + a_k\mathcal{T}_{k}Q_{k-1},\\
    & P_{k} = (1-a_k)Q_k + a_k\mathcal{T}_{k}Q_k, \\
    & Q_{k+1} = P_{k} + b_k(P_{k} - S_{k}) + c_k(Q_{k}-Q_{k-1}).
    \end{aligned}
\end{equation}
where $a_k,b_k,c_k$ are the step sizes or learning rates.
In this paper, we will mainly consider the synchronous sampling, where all the state-action pairs are updated simultaneously at each iteration round~\cite{even2003learning}. The pseudo code of the implementation is listed in Algorithm~\ref{alg:AQL}. Note that we used the notation $\mathcal{M}Q_{k}(y_{k}):=\max_{u\in U(y_k)}Q_k(y_k,u)$.

\begin{algorithm}
\caption{\label{alg:AQL}Synchronous Accelerated Q-learning}

\begin{tabbing}
\noindent\textbf{Input:} Initial action-value function $Q_0$ and $Q_{-1}=Q_0$,\\
discount factor $\gamma$, parameter $m\geq \frac{1}{\gamma}$, and maximum \\
iteration number $T$\\
\textbf{for} \= $k=0,1,2,\cdots, T-1$ \textbf{do} \\
\> $a_k = \frac{1}{k+1},\ b_{k}=k-m-1,\ c_{k}=\frac{-k^2+(m+1)k+1}{k+1}$; \\
\>\textbf{for} \= each $(x,u)\in \mathcal{X}\times U(x)$ \textbf{do}\\
\> \> Generate the next state sample $y_k\sim P(\cdot|x,u);$\\
\> \> $\mathcal{T}_{k}Q_{k-1}(x,u)= R(x,u)+\gamma\mathcal{M}Q_{k-1}(y_{k});$ \\
\> \> $\mathcal{T}_{k}Q_{k}(x,u)= R(x,u)+\gamma\mathcal{M}Q_{k}(y_{k});$ \\
\> \> $S_{k}(x,u)= (1-a_{k})Q_{k-1}(x,u)+a_{k}\mathcal{T}_{k}Q_{k-1}(x,u)$ \\
\> \> $P_{k}(x,u)= (1-a_{k})Q_{k}(x,u)+a_{k}\mathcal{T}_{k}Q_{k}(x,u)$ \\
\> \> $Q_{k+1}(x,u)= $ \= $P_{k}(x,u)+b_k\left(P_{k}(x,u)-S_{k}(x,u) \right)$ \\
\> \> \> $+c_k(Q_k(x,u)-Q_{k-1}(x,u))$ \\
\> \textbf{end for} \\
\textbf{end for}\\
\textbf{Output:} $Q_T$
\end{tabbing}
\end{algorithm}

To facilitate the analysis, we rewrite~\eqref{eq: acc tabular q learning} in a more compact way as
\begin{equation}\label{eq:compactAQL}
    \begin{aligned}
    Q_{k+1}= & (1\!-\!a_{k})Q_{k}\!+\!\left[b_{k}(1-a_{k})\!+\!c_{k}\right](Q_{k}-Q_{k-1}) \\
    & +a_{k}\left[(1+b_{k})\mathcal{T}_{k}Q_{k}-b_{k}\mathcal{T}_{k}Q_{k-1}\right].
    \end{aligned}
\end{equation}
Comparing~\eqref{eq:compactAQL} with the SpeedyQ given in~\eqref{eq:sql}, we notice that first, SpeedyQ only contains $\mathcal{T}_{k}Q_{k-1}$ in the update without explicitly using the historical information $Q_{k-1}$. This additional term in our algorithm may help attenuate possible large overshoots during the iteration. Second,~\eqref{eq:sql} simply involves $\mathcal{T}_{k}Q_k - \mathcal{T}_{k}Q_{k-1}$ as the only momentum term, while our algorithm designs this part in a more careful manner. That is, we first use two consecutive outputs of the empirical Bellman operators to update the Q-function and obtain $S_k$ and $P_k$. Intuitively, since $S_k$ and $P_k$ are derived by the update of the vanilla Q-learning, selecting $S_k-P_k$ as the additional momentum term can contribute to a better estimation of the optimal Q-function while preserving the acceleration. This intuition is also verified in our numerical results, which will be shown later. Before that, we first provide convergence analysis of the proposed algorithm.

\subsection{Convergence Rate Analysis of AQL}
Our analysis is based on the finite state-action space assumption, which is the same as in~\cite{azar2011speedy}.

\begin{assumption}\label{asp:boundSpace}
    The state space $\mathcal{X}$ and the action space $\mathcal{U}$ are finite sets with cardinalities $|\mathcal{X}|$ and $|\mathcal{U}|$, respectively. We denote $n=|\mathcal{X}|\cdot|\mathcal{U}|$. 
\end{assumption}

Our analysis starts with analyzing the errors of approximating the exact Bellman operator $\mathcal{T}$ with empirical Bellman operators $\mathcal{T}_k$. These stochastic errors and their evaluations over time are major challenges in proving convergence and deriving convergence rate. 

For convenience, we denote all $\mathcal{T}_k$ terms in~\eqref{eq:compactAQL} by
\begin{equation} 
    \mathcal{D}_k\left[Q_{k},Q_{k-1}\right]:=(1+b_k)\mathcal{T}_{k}Q_k-b_k\mathcal{T}_{k}Q_{k-1}, \label{eq:Dk}
\end{equation}
for all $k\geq 0$. Note that~\eqref{eq:Dk} is a function of all samples $\{y_1,y_2,\cdots, y_k\}$ for all station-action pair $(x,u)$ up to round $k$. Let $\mathcal{F}_k$ denote the filtration generated by the sequence of these random variables $\{y_1,y_2,\cdots, y_k\}$. 
Then if we define $\mathcal{D}\left[Q_{k},Q_{k-1}\right]$ as the conditional expectation of $\mathcal{D}_k\left[Q_{k},Q_{k-1}\right]$ given $\mathcal{F}_{k-1}$, we obtain by the definition of $\mathcal{T}$ that
\begin{align}
\mathcal{D}\left[Q_{k},Q_{k-1}\right]
&:= \mathbb{E}_{P}\left(\mathcal{D}_{k}\left[Q_{k},Q_{k-1}\right]|\mathcal{F}_{k-1}\right)\nonumber \\
&= (1+b_k)\mathcal{T}Q_{k} -b_k\mathcal{T}Q_{k-1}.\label{eq:exactD}\nonumber 
\end{align}

Now define the error between $\mathcal{D}_k$ and $\mathcal{D}$
\begin{equation}
    \epsilon_{k}:=\mathcal{D}\left[Q_{k},Q_{k-1}\right]-\mathcal{D}_{k}\left[Q_{k},Q_{k-1}\right].
\end{equation}
Clearly $\mathbb{E}_{P}\left(\epsilon_{k}|\mathcal{F}_{k-1}\right)=0$.
This shows that $\forall(x,u)\in\mathcal{X}\times U(x)$, the sequence
of estimation error $\left\{ \epsilon_{k}(x,u)\right\} _{k=0}^T$
is a martingale difference sequence with respect to the filtration $\mathcal{F}_{k}$.
In other words, if we denote 
\begin{equation}\label{eq:Ek}
    E_{k}(x,u):=\sum_{j=0}^{k}\epsilon_{j}(x,u),
\end{equation}
then $E_k$ is a martingale with respect to $\mathcal{F}_{k},$
$\forall(x,u)\in\mathcal{X}\times U(x)$ and $\forall k\geq0$. 

To proceed, we need the following assumption. 
\begin{assumption}\label{asp:boundQ}
    The Q-function is uniformly bounded throughout the learning process. That is, $\exists V_{\max}$, such that $\norm{Q_k}\leq V_{\max},\forall k\geq0$. Without loss of generality, we further let $R_{\max} +\gamma V_{\max} = V_{\max}$.
\end{assumption}

Based on Assumption~\ref{asp:boundQ}, we can further obtain the uniform bounds of $\mathcal{D}_k$ and $\epsilon_k$ as shown in the following lemma. All proofs are collected in the Appendix.
\begin{lemma}\label{lem:boundDk}
Given Assumption~\ref{asp:boundQ} and AQL as Algorithm~\ref{alg:AQL}, $\mathcal{D}_{k}\left[Q_{k},Q_{k-1}\right]$ and $\epsilon_{k}$
are uniformly bounded for all $k\geq0$. Specifically, $\exists D_{\max}>0$, s.t. $\norm{\mathcal{D}_{k}[Q_{k},Q_{k-1}]}\leq D_{\max}, \norm{\epsilon_{k}}\leq 2D_{\max}, \forall k\geq 0$.
\end{lemma}



The uniform bounds proved in Lemma~\ref{lem:boundDk} are critical in the derivation of the main theorem below.
\begin{theorem}\label{thm:main}
    Given Assumption~\ref{asp:boundSpace},~\ref{asp:boundQ} and fixing $\gamma, m$ with $\gamma m\geq 1$ in AQL as Algorithm~\ref{alg:AQL}, with probability at least $1-\delta$, the output of AQL satisfies:
    \begin{equation}\label{eq:thm}
    \begin{aligned}
        \left\Vert Q^{\star}\!-\!Q_{T}\right\Vert\!\leq\!\frac{2(\gamma R_{\max}\!+\!hV_{\max}) \!+\!D_{\max}\sqrt{8(T-m)\log \frac{2n}{\delta}}}{T(1\!-\!\gamma)},
        \end{aligned}
    \end{equation}
    where $D_{\max}$ is from Lemma~\ref{lem:boundDk} and $h=\gamma(m+1)+1$.
\end{theorem}

Combining Theorem~\ref{thm:main} and the Borel-Cantelli lemma we know $Q_T$ converges to $Q^\star$ with the rate $\mathcal{O}(\sqrt{1/T})$ almost surely. 
We further comment that the asymptotic convergence rate is comparable to that of SpeedyQ~\cite{azar2011speedy}.
Recalling our intuition for the difference of the update rules, we expect AQL can actually outperform SpeedyQ. We numerically verify this in the following.

\subsection{Numeric Comparison with SpeedyQ}
Given the analytical convergence rate derived above is also comparable to SpeedyQ. We seek to explore extra experiments to verify that our algorithm actually outperforms SpeedyQ. We emphasize that the settings in this section are consistent with those of AQL in algorithm~\ref{alg:AQL} and SpeedyQ in~\cite[Algorithm1]{azar2011speedy}. Thus the numerical results should be able to give a convincing comparison between two algorithms. Note that the choice of $m$ is not explicitly specified with only a lower bound. We try 3 different selections of $m$ and observe stable performance in convergence, which also aligns with the theoretical analysis. To evaluate the algorithms in the finite state-action space, we apply them to the popular FrozenLake grid world games. 

FrozenLake is a classic baseline problem for Q-learning. An agent controls the movement of a character in a grid world. Some tiles of the grid are walkable, and others lead to the agent falling into the water. Additionally, the movement direction of the agent is uncertain and only partially depends on the chosen direction. The agent is rewarded for finding a feasible path to a goal tile. The  environment for FrozenLake is a $4 \times 4$ grid world. We consider two sub-tasks, the FrozenLake (Fig.~\ref{fig: frozenlake comparison}) and the FrozenLake8x8 (Fig.~\ref{fig: frozenlake8 comparison}) with a bigger grid world . In both Frozenlake tasks, "S" is the safe starting point, "F" is the safe frozen surface, "H" stands for the hole that terminates the game, and "G" is the target state that comes with an immediate reward of 1. This forms a problem with state space size of $16$ ($64$ for FrozenLake8x8), action space size of $4$ and reward space $R = \{0, 1\}$.

\begin{figure}[b]
\begin{subfigure}{.22\textwidth}
  \centering
  \includegraphics[trim={4cm 0 4cm 2cm},clip,scale=0.07]{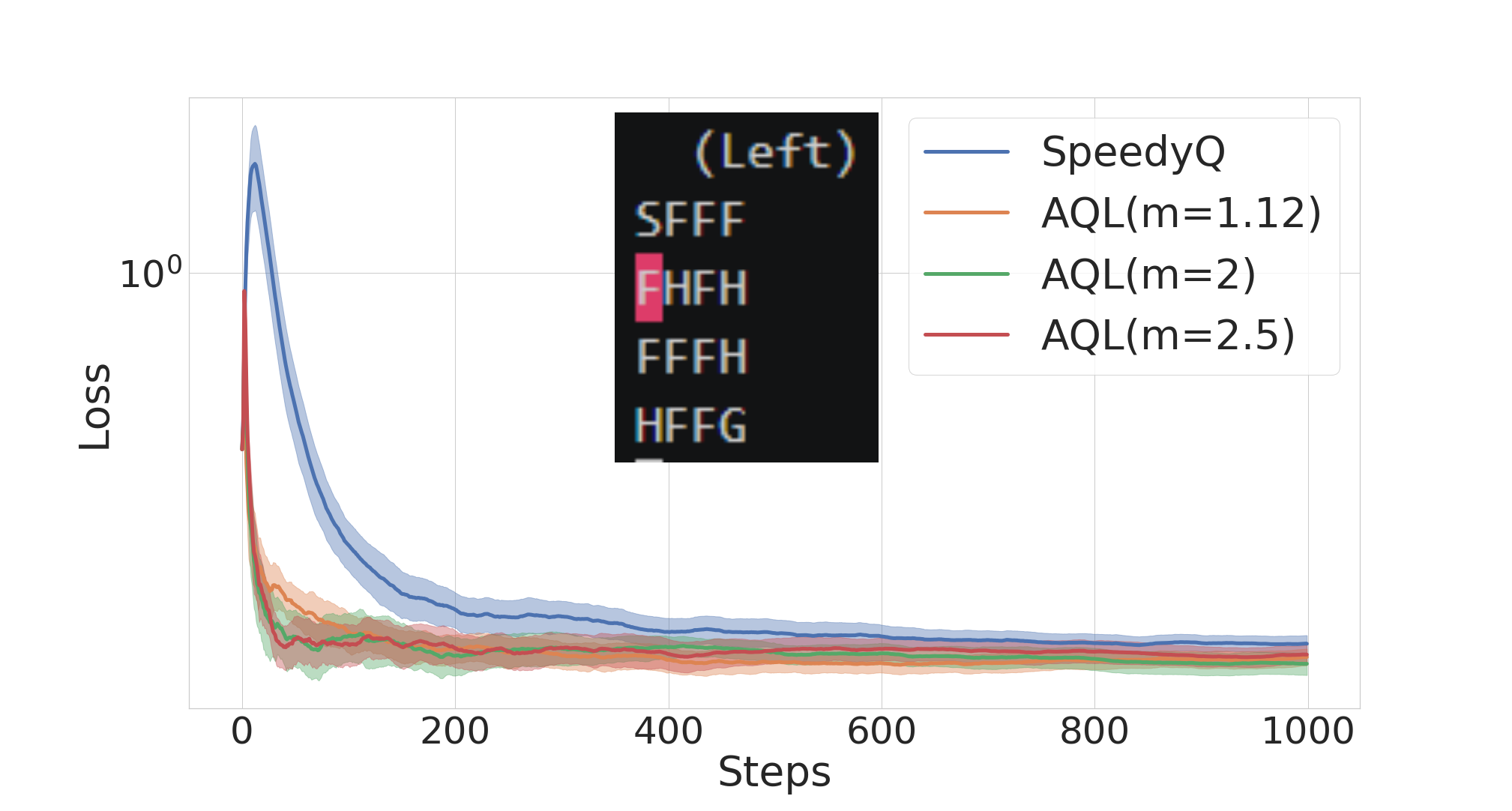}
  \caption{FrozenLake}
  \label{fig: frozenlake comparison}
\end{subfigure}
\begin{subfigure}{.22\textwidth}
  \centering
  \includegraphics[trim={4cm 0 4cm 2cm},clip,scale=0.07]{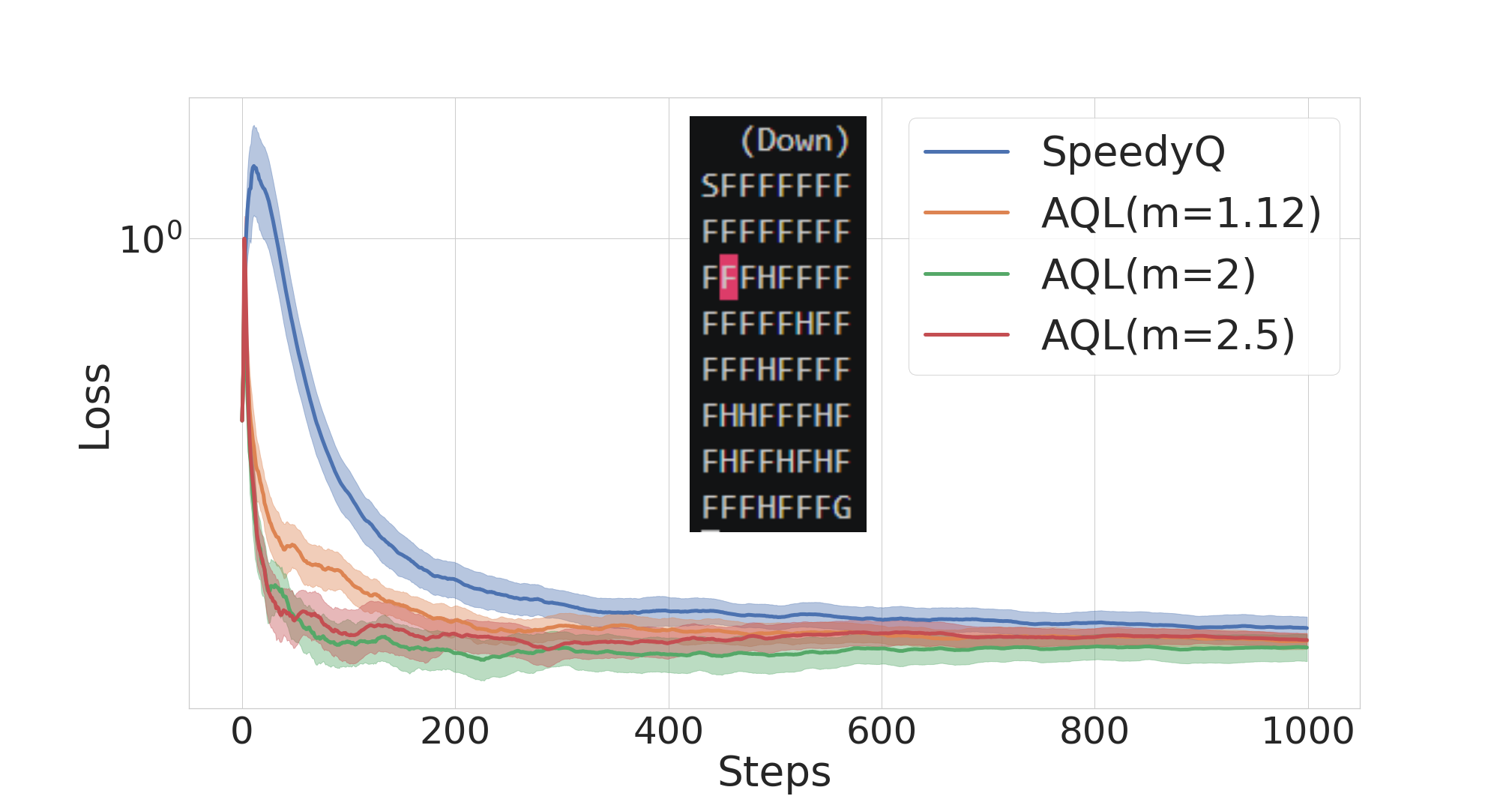}
  \caption{FrozenLake8x8}
  \label{fig: frozenlake8 comparison}
\end{subfigure}
\caption{Comparing AQL with SpeedQ.}
\label{fig:aql comp}
\vspace{-0.1in}
\end{figure}


Considering the randomness embedded in the MDP of both FrozenLake games, we evaluate the performance of each algorithm with 20 different random seeds and then illustrate the average loss and standard deviation in Fig.~\ref{fig: frozenlake comparison} and Fig.~\ref{fig: frozenlake8 comparison}. For evaluation purpose, we have access to the true transition probability, and can find the ground truth optimal Q-function $Q^{\star}$ using dynamic programming. In both games, the loss at step $k$ is then defined as $\norm{Q_k - Q^\star}$. It can be seen from the results that AQL with various choices of $m$ all can converge faster than Speedy Q-learning.

\section{GENERALIZATION TO PARAMETRIC AQL}\label{sec:ctnQ}

In this section, we generalize our acceleration scheme to the case where the state-action space is considerably large or even continuous. Numerical verification of the performance of the proposed algorithms is then provided with various tasks.

\subsection{AQL with Q-function approximation}\label{subsec:PAQL}

We consider the same MDP problem as that in Section~\ref{sec: tabularQ}, but in a continuous state-action space $\mathcal{X}\times\mathcal{U}$. In this case, it is often impossible or extremely difficult to write the Q-function as an explicit tabular function w.r.t each state-action pair, and thus the update rule of (\ref{eq:tabular q learning}) is no longer applicable.

To handle this problem, we consider a parametric function $\hat{Q}(x,u ; \theta)$ as an approximation of the Q-function. The parameter vector $\theta$ is of finite and relatively lower dimension and thus easier to implement. The approximation architectures can be rich through different choices of the function class, such as linear function approximation~\cite{bertsekas1996neuro} and neural networks~\cite{mnih2015human}. Instead of updating the estimating Q-function directly as~\eqref{eq:tabular q learning}, here we can only iteratively update the parameter $\theta$. This kind of Q-learning is referred as parametric Q-learning (PQL), which follows the update rule as
\begin{equation}\label{eq:vpql}
    \theta_{k+1}  = \theta_{k} - \alpha_{k} \Delta_k\frac{\partial}{\partial \theta_{k}}\hat{Q}_{k}(x, u; \theta_{k}),
\end{equation}
where 
\begin{equation}\label{eq:Delta}
    \Delta_k\! =\! \hat{Q}_{k}(x, u; \theta_{k})\! -\! R(x, u)\! -\! \gamma \underset{u'\in U(x')}{\max} \hat{Q}_k(x', u';\theta_{k}).
\end{equation}
Then we can generalize the proposed acceleration scheme to the vanilla PQL in~\eqref{eq:vpql}. We refer the AQL in this case as parametric AQL (PAQL) given by
\begin{equation}
    \begin{aligned}
    \label{eq: acc pq learning}
    & \xi_k = \theta_{k-1} - a_{k} \Delta_{k-1}\frac{\partial}{\partial \theta_{k-1}}\hat{Q}_{k-1}(x, u; \theta_{k-1}),\\
    & \zeta_{k} = \theta_{k} - a_{k} \Delta_k\frac{\partial}{\partial \theta_{k}}\hat{Q}_{k}(x, u; \theta_{k}), \\
    & \theta_{k+1} = \zeta_{k} + b_k(\zeta_{k} - \xi_{k}) + c_k(\theta_{k}-\theta_{k-1}).
    \end{aligned}
\end{equation}
Notice that when we take $b_k=0$, then the update~\eqref{eq: acc tabular q learning} only involves one-step historical information $\theta_{k-1}$. This applies the same idea as HB and thus is referred as HBPAQL. When taking $b_k\neq 0$, we also involve $\xi_k$ into the update, which is motivated by the idea of NAG and thus denoted as NesPAQL. 

In the following, we evaluate PAQL in two discrete-time LQR problems from the Deepmind Control Suite~\cite{tassa2018deepmind} and Atari 2600 games, where the function approximation architectures are linear functions and neural networks, respectively. 
Throughout this section, we adopt $\epsilon$-greedy~\cite{mnih2015human} and prioritized experience replay~\cite{schaul2015prioritized} for exploration and sampling, both of which are common techniques in RL with parametric approximation.

\vspace{0.1in}
\begin{figure}[b]
  \centering
  \includegraphics[scale=0.36]{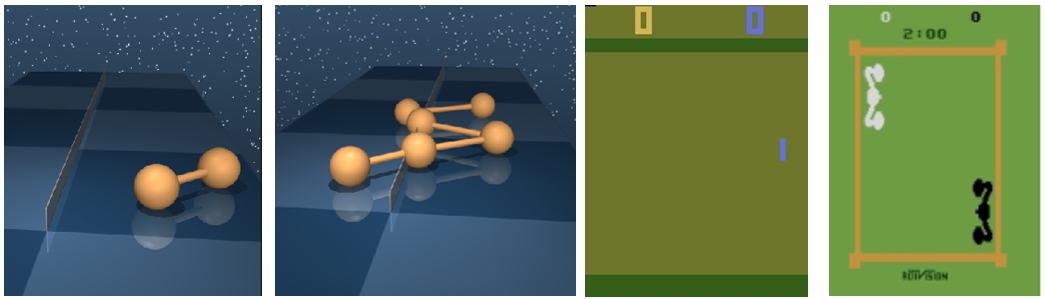}
  \caption{Testing tasks for PAQL, from left to right: LQR with 2 Masses and 1 Linear Actuators, LQR with 6 Masses and 2 Linear Actuators, Pong and Boxing from Atari 2600 Games.}
  \label{fig: env}
\end{figure}

\subsection{Linear Quadratic Regulation}\label{subsec: expLin}
The problem of infinite-horizon discrete-time LQR considers a linear system
\begin{align}
    \label{eq: linear system}
    x_{k+1} = Ax_k + Bu_k,
\end{align}
with cost function in a quadratic form as
\begin{align}
    \label{eq: lqr cost}
    J = \sum_{t=0}^{\infty}\left( x_k^T Q x_k + u_k^T R u_k + 2x_k^TNu_k\right).
\end{align}
 Let the positive definite $P$ be the unique solution to the discrete-time algebraic Riccati equation (DARE) 
\begin{equation}
    \label{eq: DARE}
    P = A^TPA\! -\! (A^TPB\! +\! N)(R\! +\! B^TPB)^{-1}(B^TPA\! +\! N^T)\! + \!Q.
\end{equation}
\normalsize
We have the optimal control as $u_k^{\star} = -K^{\star}x_k$ with
\begin{align}
    \label{eq: optimal linear solution}
    K^{\star} = (R + B^TPB)^{-1} (N^T + B^TPA).
\end{align}
Following the procedure established in section~\ref{subsec:PAQL}, we parameterize a Q-function which is linear w.r.t a matrix parameter $H$ in the form of
\begin{align}
    \label{eq: Q-function for lqr}
    Q(x, u ;H) = 
    \left[ \begin{array}{c} x \\u\end{array} \right]^T \left[ \begin{array}{c c} H_{xx} &H_{xu} \\ H_{ux} &H_{uu}\end{array} \right] \left[ \begin{array}{c} x \\u\end{array} \right].
\end{align}
The stationary linear policy corresponding to the Q-function (\ref{eq: Q-function for lqr}) satisfies $u = -Kx, K = H_{uu}^{-1}H_{ux}$. We evaluate the performance of various PAQL algorithms at each iterate $k$ with the Euclidean norm $\norm{K_k - K^{\star}}_2$.

In this section, the linear system is constructed as a coupled mass damper system with $n$ masses, serially connected through linear joints (see Fig.~\ref{fig: env}) with $m$ joints being actuated. The system has the state dimension of $2n$ with position state $x_p$ and velocity state $x_v$. The action dimension is $m$. The reward is quadratic with respect to the position and controls, i.e. $R = \frac{1}{2}x_p^TQx_p + \eta \frac{1}{2}u^TRu$ with control cost coefficient $\eta=0.1$. The system is a default RL benchmark from the Deepmind control suite~\cite{tassa2018deepmind}. We consider two sub-tasks, the "LQR\_2\_1" with $n=2$, $m=1$ and the "LQR\_6\_2" with $n=6$, $m=2$ which take 4269 and 11840 iterates respectively to converge to $K^{\star}$ through DARE.

We compare the performance of proposed the PAQL algorithms with the vanilla Q-learning in Fig~\ref{fig: lqr_2_1 comparison} and Fig.~\ref{fig: lqr_6_2 comparison}. For both tasks, we let $a_k=0.9$, $b_k=0.2$, $c_k=0.2,\forall k$ for corresponding algorithms. The learning process of DARE is also included. Direct comparison regarding the training time with DARE is not fair given that DARE requires system dynamics but Q-learning methods are model-free. In our illustration, we exclude the sampling time and consider the number of value iterations required to achieve certain level of desired performance (Table~\ref{tbl: lqr convergence}).



\begin{table}[b]
\caption{Iterates for Converging to $\norm{K_k - K^{\star}}_2 \leq 0.1$}
\label{tbl: lqr convergence}
\begin{center}
\begin{tabular}{|c||c|c|c|c|}
\hline
Task & DARE & Q-learning & HBPAQL & NesPAQL \\
\hline
LQR\_2\_1 & 769 & 515 & 229 & 205 \\
LQR\_6\_2 & 2768 & 1094 & 235 & 241 \\
\hline
\end{tabular}
\end{center}
\end{table}

\subsection{Atari 2600 games}\label{subsec: expNN}

We further evaluate the performance of PAQL with two Atari 2600 games. It is a challenging RL benchmark task that takes high-dimensional high-frequency video sequence ($dim(\mathcal{X})=84 \times 84 \times 4$) as state and real video game control keys as action. The performance for each algorithm is justified empirically by the average return of 100 trails of episodes. The Q-function is parameterized as a deep convolutional neural network. 
Hyper-parameters are set as $a_k=0.9$, $b_k=0.2$, $c_k=0.2,\forall k$. The algorithm is implemented based on the open.ai baseline, which is a set of high-quality implementations of RL algorithms. The original DQN implementation and its variants are roughly on par with scores in published papers, which mostly exceeds expert level of human play. Results are illustrated in Fig.~\ref{fig: atari comparison}. 

The Q-function is structured with millions of parameters (i.e., the weights of deep neural networks). The sampling and target learning with SGD both consume a significant amount of time and computational power. On a dual-GPU machine with the PAQL algorithm, the training for the game Pong takes $0.6$ million samples in 15 minutes. For the Boxing game to achieve the illustrated results, PAQL takes $2$ million samples in 40 minutes. On the contrast, DQN would require at least $4$ million samples to acquire similar performance.
\begin{figure}[h]
\begin{subfigure}{.45\textwidth}
  \centering
  \includegraphics[trim={0 0 0 0},clip,scale=0.15]{images/lqr_2_1.png}
  \caption{LQR\_2\_1}
  \label{fig: lqr_2_1 comparison}
\end{subfigure}
\begin{subfigure}{.45\textwidth}
  \centering
  \includegraphics[trim={0 0 0 0},clip,scale=0.15]{images/lqr_6_2.png}
  \caption{LQR\_6\_2}
  \label{fig: lqr_6_2 comparison}
\end{subfigure}
\begin{subfigure}{.45\textwidth}
  \centering
  \includegraphics[trim={0 0 0 0},clip,scale=0.15]{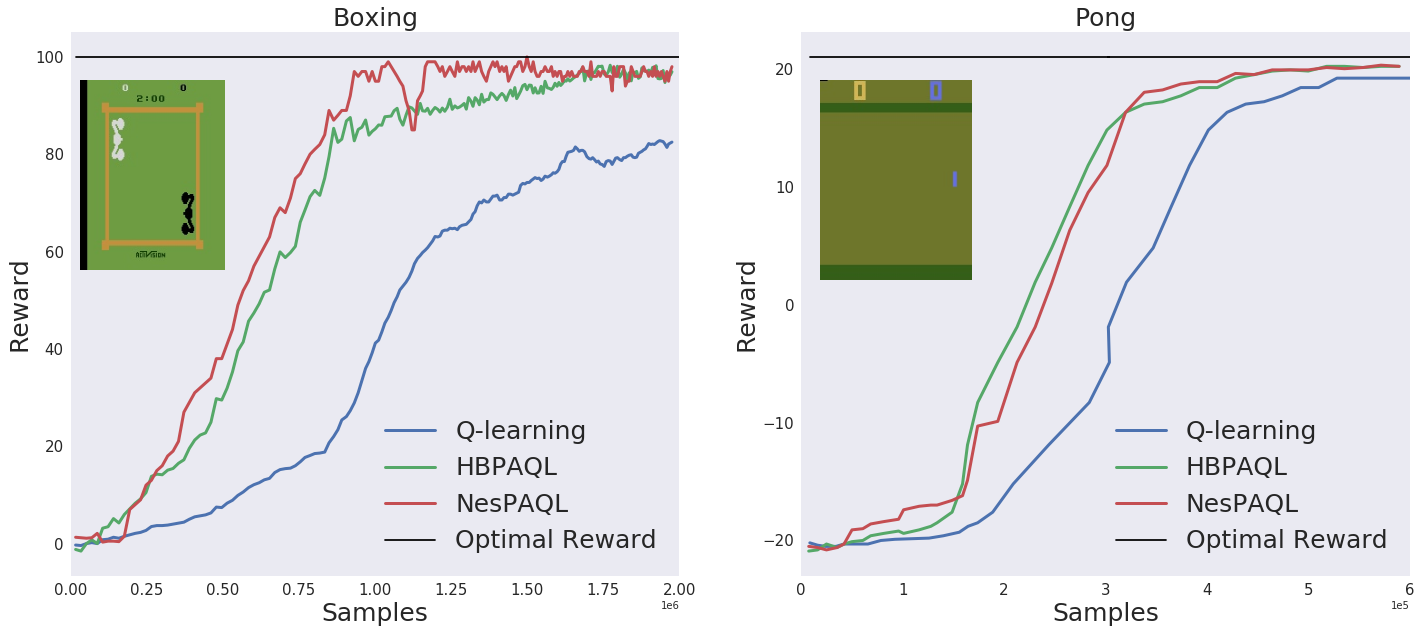}
  \caption{Atari 2600 Games}
  \label{fig: atari comparison}
\end{subfigure}
\caption{Comparing Various Methods for LQR and Atari 2600 Games.}
\label{fig:paql comp}
\vspace{-0.1in}
\end{figure}

\section{CONCLUSION}

We proposed a set of momentum-based accelerated Q-learning algorithms, which are provably converging faster than vanilla Q learning in finite state-action space if it is stable. We empirically evaluate the algorithms and verify that the proposed algorithms can accelerate the convergence in comparison to SpeedyQ and vanilla Q-learning on various challenging tasks under both finite and continuous state-action spaces settings.

Our future work includes providing theoretical guarantee of the stability of the proposed algorithm in the finite state-action space case (see Assumption~\ref{asp:boundQ}, although we never found unstable cases in our simulation). For the continuous state-action space, it is also important to study the stability of the proposed iteration scheme when using a general nonlinear approximation architecture such as neural networks. Moreover, further extensions of this work include exploring more complicated adaptive acceleration schemes to improve the convergence performance, which will be our future interests.






\bibliographystyle{plain}        
\bibliography{CDC19} 

\begin{thebibliography}{10}

\bibitem{al2007model}
Asma Al-Tamimi, Frank~L Lewis, and Murad Abu-Khalaf.
\newblock Model-free q-learning designs for linear discrete-time zero-sum games
  with application to h-infinity control.
\newblock {\em Automatica}, 43(3):473--481, 2007.

\bibitem{attouch2017rate}
Hedy Attouch, Zaki Chbani, and Hassan Riahi.
\newblock Rate of convergence of the nesterov accelerated gradient method in
  the subcritical case $\alpha$≤ 3.
\newblock {\em arXiv preprint arXiv:1706.05671}, 2017.

\bibitem{azar2011speedy}
Mohammad~Gheshlaghi Azar, Remi Munos, M~Ghavamzadaeh, and Hilbert~J Kappen.
\newblock Speedy q-learning.
\newblock 2011.

\bibitem{beck2012error}
Carolyn~L Beck and R~Srikant.
\newblock Error bounds for constant step-size q-learning.
\newblock {\em Systems \& Control Letters}, 61(12):1203--1208, 2012.

\bibitem{bertsekas1996neuro}
Dimitri~P. Bertsekas and John~N Tsitsiklis.
\newblock {\em Neuro-Dynamic Programming}, volume~5.
\newblock Athena Scientific, 1996.

\bibitem{castillo2018reinforcement}
Guillermo~A Castillo, Bowen Weng, Ayonga Hereid, and Wei Zhang.
\newblock Reinforcement learning meets hybrid zero dynamics: A case study for
  rabbit.
\newblock {\em arXiv preprint arXiv:1810.01977}, 2018.

\bibitem{devraj2017fastest}
Adithya~M Devraj and Sean~P Meyn.
\newblock Fastest convergence for q-learning.
\newblock {\em arXiv preprint arXiv:1707.03770}, 2017.

\bibitem{dozat2016incorporating}
Timothy Dozat.
\newblock Incorporating nesterov momentum into adam.
\newblock 2016.

\bibitem{even2003learning}
Eyal Even-Dar and Yishay Mansour.
\newblock Learning rates for q-learning.
\newblock {\em Journal of Machine Learning Research}, 5(Dec):1--25, 2003.

\bibitem{ghadimi2015global}
Euhanna Ghadimi, Hamid~Reza Feyzmahdavian, and Mikael Johansson.
\newblock Global convergence of the heavy-ball method for convex optimization.
\newblock In {\em 2015 European Control Conference (ECC)}, pages 310--315.
  IEEE, 2015.

\bibitem{jaakkola1994convergence}
Tommi Jaakkola, Michael~I Jordan, and Satinder~P Singh.
\newblock Convergence of stochastic iterative dynamic programming algorithms.
\newblock In {\em Advances in neural information processing systems}, pages
  703--710, 1994.

\bibitem{lewis2012reinforcement}
Frank~L Lewis, Draguna Vrabie, and Kyriakos~G Vamvoudakis.
\newblock Reinforcement learning and feedback control: Using natural decision
  methods to design optimal adaptive controllers.
\newblock {\em IEEE Control Systems Magazine}, 32(6):76--105, 2012.

\bibitem{mnih2013playing}
Volodymyr Mnih, Koray Kavukcuoglu, David Silver, Alex Graves, Ioannis
  Antonoglou, Daan Wierstra, and Martin Riedmiller.
\newblock Playing atari with deep reinforcement learning.
\newblock {\em arXiv preprint arXiv:1312.5602}, 2013.

\bibitem{mnih2015human}
Volodymyr Mnih, Koray Kavukcuoglu, David Silver, Andrei~A Rusu, Joel Veness,
  Marc~G Bellemare, Alex Graves, Martin Riedmiller, Andreas~K Fidjeland, Georg
  Ostrovski, et~al.
\newblock Human-level control through deep reinforcement learning.
\newblock {\em Nature}, 518(7540):529, 2015.

\bibitem{nesterov2013introductory}
Yurii Nesterov.
\newblock {\em Introductory lectures on convex optimization: A basic course},
  volume~87.
\newblock Springer Science \& Business Media, 2013.

\bibitem{oh2015action}
Junhyuk Oh, Xiaoxiao Guo, Honglak Lee, Richard~L Lewis, and Satinder Singh.
\newblock Action-conditional video prediction using deep networks in atari
  games.
\newblock In {\em Advances in neural information processing systems}, pages
  2863--2871, 2015.

\bibitem{polyak1964some}
Boris~T Polyak.
\newblock Some methods of speeding up the convergence of iteration methods.
\newblock {\em USSR Computational Mathematics and Mathematical Physics},
  4(5):1--17, 1964.

\bibitem{schaul2015prioritized}
Tom Schaul, John Quan, Ioannis Antonoglou, and David Silver.
\newblock Prioritized experience replay.
\newblock {\em arXiv preprint arXiv:1511.05952}, 2015.

\bibitem{schulman2017proximal}
John Schulman, Filip Wolski, Prafulla Dhariwal, Alec Radford, and Oleg Klimov.
\newblock Proximal policy optimization algorithms.
\newblock {\em arXiv preprint arXiv:1707.06347}, 2017.

\bibitem{tassa2018deepmind}
Yuval Tassa, Yotam Doron, Alistair Muldal, Tom Erez, Yazhe Li, Diego de~Las
  Casas, David Budden, Abbas Abdolmaleki, Josh Merel, Andrew Lefrancq, et~al.
\newblock Deepmind control suite.
\newblock {\em arXiv preprint arXiv:1801.00690}, 2018.

\bibitem{vamvoudakis2017q}
Kyriakos~G Vamvoudakis.
\newblock Q-learning for continuous-time linear systems: A model-free infinite
  horizon optimal control approach.
\newblock {\em Systems \& Control Letters}, 100:14--20, 2017.

\bibitem{van2016deep}
Hado Van~Hasselt, Arthur Guez, and David Silver.
\newblock Deep reinforcement learning with double q-learning.
\newblock In {\em Thirtieth AAAI Conference on Artificial Intelligence}, 2016.

\bibitem{vrabie2009adaptive}
Draguna Vrabie, O~Pastravanu, Murad Abu-Khalaf, and Frank~L Lewis.
\newblock Adaptive optimal control for continuous-time linear systems based on
  policy iteration.
\newblock {\em Automatica}, 45(2):477--484, 2009.

\bibitem{watkins1992q}
Christopher~JCH Watkins and Peter Dayan.
\newblock Q-learning.
\newblock {\em Machine learning}, 8(3-4):279--292, 1992.

\bibitem{xiong2018analytical}
Huaqing Xiong, Yuejie Chi, Bin Hu, and Wei Zhang.
\newblock Analytical convergence regions of accelerated first-order methods in
  nonconvex optimization under regularity condition.
\newblock {\em arXiv preprint arXiv:1810.03229}, 2018.

\end{thebibliography}

\appendix

\noindent\textbf{Maximal Hoeffding-Azuma Inequality:}
\begin{lemma}\label{lem:ineq}
Let $\{M_1,M_2,,\dots,M_T\}$ be a martingale difference sequence with respect to a sequence of random variables $\{X_1,X_2,,\dots,X_T\}$ (i.e. $\mathbb{E}(M_{k+1}|X_1,X_2,,\dots,X_k)=0,\forall 1\leq k\leq T$) and uniformly bounded by $\bar M>0$. If we define $S_k=\sum_{i=1}^k M_i$, then for any $\varepsilon>0$, we have
$$ \mathbb{P}\left( \underset{1\leq k\leq T}{\max}S_k>\varepsilon \right)\leq \exp\left( \frac{-\varepsilon^2}{2T\bar M^2} \right).
$$
\end{lemma}

\vspace{5mm}
\section{Proof of Lemma~\ref{lem:boundDk}}

\noindent\textbf{Proof of Lemma~\ref{lem:boundDk}:}

\begin{proof}
When $k=0$, 
\begin{align*}
\left\Vert \mathcal{D}_{0}\left[Q_{0},Q_{-1}\right]\right\Vert = & \left\Vert \mathcal{T}_{0}Q_{0}\right\Vert \leq\left\Vert R\right\Vert +\gamma\left\Vert \mathcal{M}Q_{0}(y_{0})\right\Vert \\
\leq & R_{\max}+\gamma V_{\max}=V_{\max}:=\bar B_0
\end{align*}

Now, considering $k\geq1$ we have 
\begin{align*}
&\left\Vert \mathcal{D}_{k}\left[Q_{k},Q_{k-1}\right]\right\Vert \\
\leq & \left\Vert R\right\Vert +\gamma\Vert(1+b_k)\mathcal{M}Q_{k}-b_k\mathcal{M}Q_{k-1}\Vert\\
= & R_{\max}+\gamma\Vert(1+b_k)\mathcal{M}\bigl(Q_{k-1}\\
 & -\alpha_{k-1}Q_{k-2}+\alpha_{k-1}\mathcal{D}_{k-1}\left[Q_{k-1},Q_{k-2}\right]\bigr)\\
 & -b_k\mathcal{M}Q_{k-1}\Vert\\
\leq & R_{\max}+\gamma\norm{Q_{k-1}} + \gamma |1+b_k| a_{k-1}\norm{Q_{k-2}}\\
&+\gamma |1+b_k|\alpha_{k-1}\norm{\mathcal{D}_{k-1}\left[Q_{k-1},Q_{k-2}\right]},
\end{align*}
where the first inequality follows from the triangle inequality and the second follows due to the triangle inequality and the definition of the infinity norm. 

By the choice of the hyper-parameters in Algorithm~\ref{alg:AQL}, we know $|1+b_k| a_{k-1}=\frac{|k-m|}{k}$. Then we consider two cases: $1\leq k<\frac{m}{2}$ and $k\geq \frac{m}{2}$. The first case only contains finite steps, and in the second case, we can simply bound $|1+b_k| a_{k-1}=\frac{|k-m|}{k}\leq 1$.

When $1\leq k<\frac{m}{2}$, i.e. when $1<|1+b_k| a_{k-1}=\frac{|k-m|}{k}\leq m$, we have
\begin{align*}
    &\left\Vert \mathcal{D}_{k}\left[Q_{k},Q_{k-1}\right]\right\Vert \\
\leq & R_{\max}+\gamma\norm{Q_{k-1}} + \gamma |1+b_k| a_{k-1}\norm{Q_{k-2}}\\
&+\gamma |1+b_k|\alpha_{k-1}\norm{\mathcal{D}_{k-1}\left[Q_{k-1},Q_{k-2}\right]}\\
= & R_{\max}+\gamma\norm{Q_{k-1}} + \gamma \frac{m-k}{k}\norm{Q_{k-2}}\\
&+\gamma \frac{m-k}{k}\norm{\mathcal{D}_{k-1}\left[Q_{k-1},Q_{k-2}\right]}\\
\leq & R_{\max}+\gamma V_{\max} +\gamma mV_{\max} +\gamma m\norm{\mathcal{D}_{k-1}\left[Q_{k-1},Q_{k-2}\right]}\\
\leq & (1+\gamma m) V_{\max}\sum_{i=0}^{k-1}(\gamma m)^i + (\gamma m)^k\norm{\mathcal{D}_{0}\left[Q_{k-1},Q_{k-2}\right]}\\
\leq & \frac{(1+\gamma m)((\gamma m)^{k-1}-1) V_{\max}}{\gamma m-1} + (\gamma m)^k V_{\max},
\end{align*}
where the third inequality follows from Assumption~\ref{asp:boundQ}. Since $\gamma m\geq 1$, in this case $D_k$ can be uniformly bounded as
$$\begin{aligned} &\left\Vert \mathcal{D}_{k}\left[Q_{k},Q_{k-1}\right]\right\Vert\\
&\leq \left(\frac{(1+\gamma m)((\gamma m)^{\lfloor m/2 \rfloor-1}-1)}{\gamma m-1}+(\gamma m)^{\lfloor m/2 \rfloor} \right)V_{\max}:=\bar B_1,
\end{aligned}
$$
where $\lfloor * \rfloor$ denotes the largest integer that is no larger than $*$.

Then we consider when $k\geq \frac{m}{2}$, i.e. when $|1+b_k| a_{k-1}=\frac{|k-m|}{k}\leq 1$:
\begin{equation}\label{eq:pf}
\begin{aligned}
&\left\Vert \mathcal{D}_{k}\left[Q_{k},Q_{k-1}\right]\right\Vert \\
\leq & R_{\max}+\gamma\norm{Q_{k-1}} + \gamma |1+b_k| a_{k-1}\norm{Q_{k-2}}\\
&+\gamma |1+b_k|\alpha_{k-1}\norm{\mathcal{D}_{k-1}\left[Q_{k-1},Q_{k-2}\right]}\\
\leq & R_{\max}+\gamma\norm{Q_{k-1}} + \gamma \norm{Q_{k-2}},
\end{aligned}
\end{equation}
Notice that the second inequality follows due to $|1+b_k| a_{k-1}=\frac{|k-m|}{k}\leq 1$, which is the main difference from the previous case. Then we can further bound
\begin{equation}\label{eq:pf4}
\begin{aligned}
&\left\Vert \mathcal{D}_{k}\left[Q_{k},Q_{k-1}\right]\right\Vert \\
\leq & R_{\max}+2\gamma V_{\max}+\gamma\left\Vert \mathcal{D}_{k-1}\left[Q_{k-1},Q_{k-2}\right]\right\Vert \\
\leq & R_{\max}\sum_{i=0}^{k-\lfloor m/2 \rfloor}\gamma^{i}+2V_{\max}\sum_{i=1}^{k-\lfloor m/2 \rfloor}\gamma^{i}\\
&+ \gamma^{k-\lfloor m/2 \rfloor}\left\Vert \mathcal{D}_{\lfloor m/2 \rfloor}\left[Q_{0},Q_{-1}\right]\right\Vert \\
\leq & \frac{\left(R_{\max}+2\gamma V_{\max}\right)}{1-\gamma}+ \bar B_1\\
= & \frac{1+\gamma}{1-\gamma} V_{\max}+ \bar B_1 :=\bar B_2,
\end{aligned}
\end{equation}
where the first inequality follows from Assumption~\ref{asp:boundQ}.
Observe that clearly $\bar B_0 < \bar B_1 < \bar B_2$. Then we can uniformly bound $\norm{D_k}$ as 
$$\left\Vert \mathcal{D}_{k}\left[Q_{k},Q_{k-1}\right]\right\Vert\leq \bar B_2:=D_{\max},\ \forall k\geq 0.$$

The bound on $\epsilon_{k}$ follows directly from its definition as
$$\begin{aligned}
\left\Vert \epsilon_{k}\right\Vert &=\left\Vert \mathbb{E}_{P}\left(\mathcal{D}_{k}\left[Q_{k},Q_{k-1}\right](x,u)|\mathcal{F}_{k-1}\right)-\mathcal{D}_{k}\left[Q_{k},Q_{k-1}\right]\right\Vert\\
&\leq2\left\Vert \mathcal{D}_{k}\left[Q_{k},Q_{k-1}\right]\right\Vert\leq 2D_{\max}.
\end{aligned}$$
Thus we conclude our proof.
\end{proof}

\noindent\textbf{Proof of Theorem~\ref{thm:main}:}

We first prove two lemmas that will lead to the main results.
The first lemma derives the dynamics of $Q_k$ in terms of $E_k$, which will be handy later.
\begin{lemma}\label{lem:dynQ}
For any $k\geq1$ and given AQL as in Algorithm~\ref{alg:AQL}, we have 
\begin{equation}\label{eq:2termIter}
\begin{aligned}
Q_{k}=&\frac{1}{k}(Q_{k-1}-Q_{0}+(k-m-1)\mathcal{T}Q_{k-1})\\
&+\frac{1}{k}((m+1)\mathcal{T}Q_{0}-E_{k-1})
\end{aligned}
\end{equation}
\end{lemma}
\begin{proof}
First we rewrite~\eqref{eq:compactAQL} by the definition of $D_k$ and $\epsilon_k$ as
\begin{equation}\label{eq:oneLineIteration}
    \begin{aligned}
    Q_{k+1}\!=\! & (1\!-\!a_{k})Q_{k}\!+\!\left[b_{k}(1-a_{k})\!+\!c_{k}\right](Q_{k}-Q_{k-1}) \\
    & +a_{k}\left[(1+b_{k})\mathcal{T}_{k}Q_{k}-b_{k}\mathcal{T}_{k}Q_{k-1}\right] \\
    = & (1\!-\!a_{k})Q_{k}\!+\!\left[b_{k}(1-a_{k})\!+\!c_{k}\right](Q_{k}-Q_{k-1}) \\
    & +a_{k}\mathcal{D}_{k}\left[Q_{k},Q_{k-1}\right]\\
    = & (1\!-\!a_{k})Q_{k}\!+\!\left[b_{k}(1-a_{k})\!+\!c_{k}\right](Q_{k}-Q_{k-1}) \\
    & +a_{k}\left(\mathcal{D}\left[Q_{k},Q_{k-1}\right]-\epsilon_k\right).
    \end{aligned}
\end{equation}

Then we prove the lemma by plugging in the choice of the hyper-parameters and using induction. For $k=1$, $Q_{1}=\mathcal{T}_{1}Q_{0}=\mathcal{T}Q_{0}-E_{1}$, Thus (\ref{eq:2termIter}) holds when $k=1$. Now under the assumption that (\ref{eq:2termIter})
holds for $k$ we prove it also holds for $k+1$.
\begin{align*}
&Q_{k+1}\\
= & \frac{1}{k+1}Q_{k}-\frac{1}{k+1}Q_{k-1}+\frac{k}{k+1}Q_{k}\\
 & +\frac{1}{k+1}\left[(k-m)\mathcal{T}_{k}Q_{k}-(k-m-1)\mathcal{T}_{k}Q_{k-1}\right]\\
= & \frac{1}{k+1}Q_{k}-\frac{1}{k+1}Q_{k-1}+\frac{1}{k+1}(Q_{k-1}-Q_{0}\\
 & +(k-m-1)\mathcal{T}Q_{k-1}+(m+1)\mathcal{T}Q_{0}-E_{k-1})\\
 & +\frac{1}{k+1}\left[(k-m)\mathcal{T}_{k}Q_{k}-(k-m-1)\mathcal{T}_{k}Q_{k-1}\right]\\
= & \frac{1}{k+1}Q_{k}-\frac{1}{k+1}Q_{k-1}+\frac{1}{k+1}(Q_{k-1}-Q_{0}\\
 & +(k-m-1)\mathcal{T}Q_{k-1}+(m+1)\mathcal{T}Q_{0}-E_{k-1})\\
 & +\frac{1}{k+1}\left[(k-m)\mathcal{T}Q_{k}-(k-m-1)\mathcal{T}Q_{k-1}-\epsilon_{k}\right]\\
= & \frac{1}{k+1}(Q_{k}-Q_{0}+(k-m)\mathcal{T}Q_{k}+(m+1)\mathcal{T}Q_{0}-E_{k}),
\end{align*}
which shows (\ref{eq:2termIter}) holds for $k+1$, and therefore
it holds for all $k\geq1$.
\end{proof}

The second lemma derives the propagation of the errors $\epsilon_k$ in the process of $Q$ function iteration, which can be proved conveniently using Lemma~\ref{lem:dynQ}.
\begin{lemma}\label{lem:errorProp}
Given Assumption~\ref{asp:boundQ} and fixing $\gamma,m$ with $\gamma m\geq 1$ in AQL as Algorithm~\ref{alg:AQL}, for all $k\geq m+1$, we have
\begin{equation}\label{eq:pf5}
    \left\Vert Q^{\star}\!-\!Q_{k}\right\Vert\leq 2\frac{\gamma R_{\max}+hV_{\max}}{k(1-\gamma)}\!+\!\frac{1}{k}\sum_{i=0}^{k-m-1}\gamma^{i}\Vert E_{k-i}\Vert,
\end{equation}
where $h=\gamma(m+1)+1$.
\end{lemma}
\begin{proof}
For $k\geq m+1$, expand $Q_k$ using (\ref{eq:2termIter}) in Lemma~\ref{lem:dynQ}, and we have
\begin{align*}
&\left\Vert Q^{\star}-Q_{k}\right\Vert \\
= & \frac{1}{k}\Vert Q_{k-1}-Q_{0}+(k-m-1)(\mathcal{T}Q^{\star}-\mathcal{T}Q_{k-1})\\
 & +(m+1)(\mathcal{T}Q^{\star}-\mathcal{T}Q_{0})+E_{k}\Vert\\
\leq & \frac{\gamma(k-m-1)+1}{k}\Vert Q^{\star}-Q_{k-1}\Vert\\
 & +\frac{\gamma(m+1)+1}{k}\Vert Q^{\star}-Q_{0}\Vert+\frac{\Vert E_{k}\Vert}{k}\\
\leq & \frac{\gamma(k-1)}{k}\Vert Q^{\star}-Q_{k-1}\Vert+\frac{2h}{k}V_{\max}+\frac{\Vert E_{k}\Vert}{k}\\
\leq & \frac{\gamma^{k-m}}{k}\Vert Q^{\star}-Q_{m}\Vert+\frac{2hV_{\max}}{k}\sum_{i=0}^{k-m-1}\gamma^{i}\\
 & +\sum_{i=0}^{k-m-1}\frac{\gamma^{i}}{k}\Vert E_{k-i}\Vert\\
\leq & 2\frac{\gamma R_{\max}+hV_{\max}}{k(1-\gamma)}+\frac{1}{k}\sum_{i=0}^{k-m-1}\gamma^{i}\Vert E_{k-i}\Vert,
\end{align*}
where the first inequality follows from the triangle inequality and the contraction property~\eqref{eq:Contraction}, and the second inequality holds due to $\gamma m\geq 1$ and Assumption~\ref{asp:boundQ} and $h=\gamma(m+1)+1$.
\end{proof}

Now we are ready to prove the main results of Theorem~\ref{thm:main}. The proof builds on the results of Lemma~\ref{lem:errorProp} and makes uses of the Maximal Hoeffding-Azuma Inequality (see Lemma~\ref{lem:ineq}).

\begin{proof}
 Plugging $k=T$ in~\eqref{eq:pf5} in Lemma~\ref{lem:errorProp} and obtain
$$
    \left\Vert Q^{\star}-Q_{T}\right\Vert\leq 2\frac{\gamma R_{\max}+hV_{\max}}{T(1-\gamma)}+\frac{1}{T}\sum_{i=0}^{T-m-1}\gamma^{i}\Vert E_{T-i}\Vert.
$$
It suffices to bound the second term. Observe that 
\begin{equation}\label{eq:pf3}
    \begin{aligned}
    \frac{1}{T}\sum_{i=0}^{T-m-1}\gamma^{i}\Vert E_{T-i}\Vert
    &\leq\frac{1}{T}\sum_{i=0}^{T-m-1}\gamma^{i}\underset{0\leq i\leq T-m-1}{\max}\norm{E_{T-i}}\\
    &\leq \frac{\max_{0\leq i\leq T-m-1}\norm{E_{T-i}}}{(1-\gamma)T}.
    \end{aligned}
\end{equation}

In remains to bound $\max_{0\leq i\leq T-m-1}\norm{E_{T-i}}$. For the sake of convenience, we denote $K=T-m-1$. Notice that $\max_{0\leq i\leq K}\norm{E_{T-i}}=\max_{(x,u)} \max_{0\leq i\leq K}\lvert{E_{T-i}(x,u)}\rvert$. For a given 
$(x,u)$ and $\varepsilon>0$, we have
\begin{equation}\label{eq:pf1}
    \begin{aligned}
    &\mathbb{P}\left( \underset{0\leq i\leq K}{\max}\lvert E_{T-i}(x,u)\rvert>\varepsilon \right)\\
    =&\mathbb{P}\left( \left\{\underset{0\leq i\leq K}{\max}( E_{T-i}(x,u)) >\varepsilon\right\}\right.\\
    &\left.\bigcup \left\{\underset{0\leq i\leq K}{\max}( -E_{T-i}(x,u)) >\varepsilon\right\} \right)\\
    =&\mathbb{P}\left( \underset{0\leq i\leq K}{\max}( E_{T-i}(x,u))>\varepsilon \right) +\\ &\mathbb{P}\left( \underset{0\leq i\leq K}{\max}( -E_{T-i}(x,u))>\varepsilon \right),
    \end{aligned}
\end{equation}
where $D_{\max}$ is derived in Lemma~\ref{lem:boundDk}.
Since $\{\epsilon_k(x,u)\}_{k\geq 0}$ is a martingale difference sequence w.r.t the filtration $\mathcal{F}_k$ as defined previously, we can apply the Maximal Hoeffding-Azuma inequality (see Lemma~\ref{lem:ineq}) as
\begin{equation}\nonumber
    \begin{aligned}
    \mathbb{P}\left( \underset{0\leq i\leq K}{\max}( E_{T-i}(x,u))>\varepsilon\right) &\leq \exp\left( \frac{-\varepsilon^2}{8(K+1)D_{\max}^2} \right)\\
    \mathbb{P}\left( \underset{0\leq i\leq K}{\max}( -E_{T-i}(x,u))>\varepsilon\right) &\leq \exp\left( \frac{-\varepsilon^2}{8(K+1)D_{\max}^2} \right).
    \end{aligned}
\end{equation}
Then we can further bound~\eqref{eq:pf1} as
$$ \mathbb{P}\left( \underset{0\leq i\leq K}{\max}\lvert E_{T-i}(x,u)\rvert>\varepsilon \right)\leq 2\exp\left( \frac{-\varepsilon^2}{8(K+1)D_{\max}^2} \right).
$$
Since we consider a finite state-action space where the size of state-action pairs is bounded by $n$ as Assumption~\ref{asp:boundSpace}, we can eventually use the union bound to obtain
$$ \mathbb{P}\left( \underset{0\leq i\leq K}{\max}\lVert E_{T-i}\rVert>\varepsilon \right)\leq 2n\exp\left( \frac{-\varepsilon^2}{8(K+1)D_{\max}^2} \right).
$$
By letting $\delta = 2n\exp\left( \frac{-\varepsilon^2}{8TD_{\max}^2} \right)$ we have
$$ \mathbb{P}\left( \underset{0\leq i\leq K}{\max}\lVert E_{T-i}\rVert\leq D_{\max}\sqrt{8(K+1)\log \frac{2n}{\delta}} \right)\geq 1-\delta,
$$
where $K=T-m-1$.
By plugging the above probability bound in~\eqref{eq:pf3} we conclude our results.
\end{proof}



\end{document}